\documentclass[aps,pra,reprint,superscriptaddress,longbibliography,nofootinbib]{revtex4-1}
\usepackage{amsfonts}
\usepackage{amsmath}
\usepackage{amssymb}
\usepackage{amsthm}
\usepackage{mathtools}
\usepackage{enumerate}
\usepackage{ulem}
\usepackage{bm}
\usepackage{graphicx}
\usepackage{nicefrac}
\usepackage{color}
\usepackage[all]{xy}
\usepackage{hyperref}
\usepackage{bbm}
\usepackage{tikz}
\usetikzlibrary{positioning}
\usepackage{MnSymbol}
\usepackage[caption=false]{subfig}

\usepackage{soul,xcolor}
\setstcolor{purple}

\newtheorem{theorem}{Theorem}
\newtheorem*{theorem*}{Theorem}
\newtheorem{lemma}{Lemma}
\newtheorem{corollary}{Corollary}
\newtheorem{definition}{Definition}

\newcommand{\zz}{\mathbb Z}
\newcommand{\calC}{\mathcal{C}}

\newcommand{\ket}[1]{\ensuremath{\left|#1\right\rangle}}

\begin{document}
\title{Contextuality as a resource for measurement-based quantum computation beyond qubits}
\author{Markus Frembs}
\affiliation{Department of Physics, Imperial College London, London SW7 2AZ, United Kingdom}
\author{Sam Roberts}
\author{Stephen D. Bartlett}
\affiliation{Centre for Engineered Quantum Systems, School of Physics, The University of Sydney, Sydney, Australia}

\begin{abstract}
Contextuality---the obstruction to describing quantum mechanics in a classical statistical way---has been proposed as a resource that powers quantum computing.  The measurement-based model provides a concrete manifestation of contextuality as a computational resource, as follows.  If local measurements on a multi-qubit state can be used to evaluate non-linear boolean functions with only linear control processing, then this computation constitutes a proof of strong contextuality---the possible local measurement outcomes cannot all be pre-assigned.  However, this connection is restricted to the special case when the local measured systems are \textit{qubits}, which have unusual properties from the perspective of contextuality.  A single qubit cannot allow for a proof of contextuality, unlike higher-dimensional systems, and multiple qubits can allow for state-independent contextuality with only Pauli observables, again unlike higher-dimensional generalisations.  Here we identify precisely that strong non-locality is necessary in a qudit measurement-based computation that evaluates high-degree polynomial functions with only linear control.  We introduce the concept of \textit{local universality}, which places a bound on the space of output functions accessible under the constraint of single-qudit measurements. Thus, the partition of a physical system into subsystems plays a crucial role for the increase in computational power.  A prominent feature of our setting is that the enabling resources for qubit and qudit measurement-based computations are of the same underlying nature, avoiding the pathologies associated with qubit contextuality.
\end{abstract}

\maketitle


\section{Introduction}

Computers that exploit quantum phenomena are believed to be more powerful than those obeying classical rules.  Many features of quantum theory have been proposed as the origin of this supposed quantum computational power:  entanglement, superposition, and exponential scaling of Hilbert spaces to name a few.  Recently, \textit{contextuality} has been investigated in a variety of scenarios as a potential resource for quantum computation \cite{Raussendorf2009,howard2014contextuality,bermejo2017contextuality,delfosse2015wigner,raussendorf2017contextuality,karanjai2018contextuality,raussendorf2016cohomological,okay2017topological,veitch2014resource,CataniBrowne2017,MansfieldKashefi2018,deSilva2018,Galvao2005,PashayanWallmanBartlett2015}. Contextuality can be thought of as the impossibility of assigning pre-determined outcomes to all potential measurements of a quantum system, independent of their measurement context~\cite{KochenSpecker,mermin1990simple,mermin1993hidden,peres1991two}. This fundamental but peculiar property of quantum systems can exhibit in many ways, most famously in allowing quantum systems to circumvent constraints on classical correlations, leading to strong non-locality (to be defined shortly).

While the majority of recent research into contextuality as a resource for quantum computation sits in the framework of the circuit model, perhaps the most striking results in this direction arise in the measurement-based model of quantum computation (MBQC)~\cite{RaussendorfBrowneBriegel2003}.  Anders and Browne~\cite{AndersBrowne2009} showed that a simple control computer limited to evaluating linear boolean functions can be boosted in power, to one that can evaluate general (non-linear) functions, when given the measurement outcomes on a resource state that constitutes a proof of contextuality.  Their key example was Mermin's simplified GHZ paradox~\cite{mermin1990simple} (a common proof of contextuality), where linear control of the local measurement settings allows for the evaluation of a non-linear NAND gate.  Subsequently, Raussendorf~\cite{Raussendorf2009} extended these results, proving that the computation of a non-linear function (from measurement outcomes with linear pre- and post-processing) implies the impossibility of non-contextual assignments to the single qubit observables. 

These previous results are quite strongly dependent on qubits (two-dimensional quantum systems) being the elementary measured systems.  While qubits are standard in quantum computation, they are unusual from the perspective of contextuality.  For example, an individual qubit is unique as a quantum system that cannot be used to prove Kochen-Specker contextuality~\cite{mermin1993hidden,KochenSpecker}.  On the other hand, Pauli observables on multi-qubit systems (such as in Mermin's simplified Peres and GHZ paradoxes~\cite{mermin1990simple}) allow for state-independent proofs of contextuality, whereas the natural generalisation of these thought experiments to higher-dimensional systems are non-contextual~\cite{Gross2006,veitch2012negative}.  As highlighted in Anders and Browne, deriving a generalisation of their result for measurements of higher-dimensional systems is not straightforward.  In particular, using qudits (higher-dimensional quantum systems) has the potential to confuse the role of contextuality within an individual system with the exhibition of contextuality as a non-local (in Bell's sense) correlation between quantum systems.  In addition, the Mermin-style proof of strong contextuality which serves as the key non-linear example in the qubit case does not straightforwardly generalise to qudits (with generalised Pauli observables).  The key prior research in this direction is by Hoban \textit{et al.}~\cite{HobanWallmanBrowne2011}, where the qudit case was considered from the perspective of Bell inequalities for multi-qudit systems.  It was shown that evaluation of a sufficiently high order polynomial function on a multi-qudit system constituted a proof of strong contextuality.

In this paper, we characterise the role of contextuality and non-locality within a general framework of both qubit and qudit MBQC. We give examples of non-contextual qudit MBQCs (with local dimension $d\geq3$) that evaluate non-linear functions, a result that sits in stark contrast to the qubit case; these examples show that a naive generalisation of Raussendorf's result is not possible. We then consider the space of functions that can be evaluated using MBQCs that admit pre-determined local value assignments to each measured observable. For a polynomial function (in several variables) the existence of a value assignment to local observables places a restriction on its combined degree.  In particular, we prove that the evaluation of a polynomial of sufficiently high degree gives a proof of strong contextuality, reproducing the result of Ref.~\cite{HobanWallmanBrowne2011}, in a way that emphasises the distinctive role of contextuality in local vs global systems and correlations.  A key ingredient in the argument is the notion of \textit{local universality}, which captures both the computational power of MBQCs that admit local value assignments, and exposes the importance of non-locality---that grouping qubits or qudits together to form larger degrees of freedom can lead to additional computational power.  While our results are general (and can apply to post-quantum theories more broadly), we give particular emphasis to the stabilizer-based MBQCs that are most commonly considered in the quantum computing literature.  

The paper is organised as follows. In Sec.~\ref{sec: The Setup} we present a general framework for measurement-based computations (MBC) following Anders and Browne \cite{AndersBrowne2009}, and focus on the important case of MBQC on both qubits and qudits.  In Sec.~\ref{sec: Examples and puzzles}, we highlight some of the problems in generalising the qubit-based results. For instance, we give explicit examples that show how non-linear functions can be computed even in non-contextual qudit MBQC.  In Sec.~\ref{secLocalUniv}, we review some results on functions in finite fields, and derive a description of subspaces in the space of polynomial functions invariant under linear pre- and post-processing. We also introduce the crucial notion of local universality, and prove that MBQC within the stabilizer formalism is itself locally universal, i.e.,~it allows to implement arbitrary functions on individual qudits. The classification of function spaces together with local universality eventually leads to a generalisation of Raussendorf's theorem from qubits to qudits, given in Sec.~\ref{sec: Non-locality in MBQC}.


\section{The Setup}\label{sec: The Setup}

We start by introducing the notion of computation we consider, called $ld$-MBC for ``Measurement-Based Computation with $\mathbb{Z}_d$-linear classical processing''.  This notion fits within the computational framework first introduced in Anders and Browne~\cite{AndersBrowne2009} to study the computational power of correlated resources in a general setting that includes MBQC. For MBQC it is natural to choose the input and output alphabet to be a cyclic group and we restrict our definition of $ld$-MBC in this way. While we are primarily concerned with generalising the results of Raussendorf~\cite{Raussendorf2009} to qudits of prime power dimension, many of our results are not restricted to this particular implementation and hold more generally for any resource with non-local correlations. As such we briefly review the Anders and Browne setup before discussing the special case of MBQC.

A general MBC consists of two components: a correlated resource, and a control computer with restricted computational power. The correlated resource consists of $N$ local parties, each of which is allowed to exchange classical information with the control computer once. No communication between parties is allowed during the computation, and the correlations in their output are entirely due to interactions prior to the computation. During the exchange with the control computer, each party receives an element of $\zz_k$ from the control computer (called the measurement setting), and returns an outcome that is an element of $\zz_l$ (called the measurement outcome). The control computer combines the measurement outcomes in a linear way to produce the computational output.

We are interested in the case where inputs and output are of fixed dimension $d$, such that $d=k=l= p^r$ for $p$ prime, $r \in \mathbb{N}$. In this case we refer to the system as an $ld$-MBC, which is defined more precisely as follows.

\begin{definition}\label{defldMBC}
A $ld$-MBC with classical input $\mathbf{i}$ and classical output $o(\mathbf{i})$ consists of $N$ parties, each of which receives an input $q_k \in \zz_d$ from the control computer, and returns an outcome $s_k\in \zz_d$, for $k = 1,\ldots N$. The inputs and computational output satisfy the following conditions:
\begin{enumerate}
\item The computational output $o(\mathbf{i}) \in \mathbb{Z}_d$ is a linear function of the local measurement outcomes $\mathbf{s} = (s_1,\cdots,s_N)^\intercal$,
\begin{equation}\label{eq: MBC linear post-processing}
o(\mathbf{i}) = Z\mathbf{s} + s_0 \quad \mathrm{mod}\ d,
\end{equation}
for $s_0 \in \mathbb{Z}_d$ and $Z\in\text{Mat}_n(\zz_d)$.
\item The choice of measurements $\mathbf{q} = (q_1,\cdots,q_N)^\intercal$ is related to the measurement outcomes $\mathbf{s}$ and the $\mathbb{Z}_d$-valued classical input $\mathbf{i} = (i_1,\cdots,i_n)^\intercal$ via
\begin{equation}\label{eq: MBC linear pre-processing}
\mathbf{q} = T\mathbf{s} + Q\mathbf{i} \quad \mathrm{mod}\ d,
\end{equation}
for some $T,Q\in \text{Mat}_n(\zz_d)$.
\item For a suitable ordering of the parties $1,\cdots,n$ the matrix $T$ in Eq.~(\ref{eq: MBC linear pre-processing}) is lower triangular with vanishing diagonal. If $T = 0$ the $ld$-MBC is called temporally flat.
\end{enumerate}
\end{definition}

In the above, $\text{Mat}_n(\zz_d)$ denotes the space of $n\times n$ matrices with entries in $\zz_d$. We will mostly be concerned with a special implementation of this setup known as $ld$-MBQC, where the correlated resource is given by a quantum state, and the exchanged information with each party are local measurement settings and outcomes. In particular, a general $ld$-MBQC consists of two components:  a correlated quantum resource, and a control computer with restricted computational power. The quantum resource consists of $N$ local parties, each of which contains a quantum system of dimension $d$ (i.e., a qudit) and a measurement device. For each party there is a choice of $d$ measurement settings, each with $d$ measurement outcomes.

Each party exchanges data with the control computer once. Namely, each party receives a measurement setting $q_k \in \zz_d$ from the control computer to determine the choice of measurement $M_k(q_k)$, and returns the measurement outcome $m_k(q_k) \in \zz_d$, where here and thereafter $k \in \{1,\ldots, N\}$ labels the party. We assume that the eigenvalues of each $M_k(q_k)$ are of the form $\omega^{z}$ for $\omega = e^{\frac{2 \pi i}{d}}$ a $d$th root of unity and $z \in \zz_d$.  Note that such operators are not Hermitian, but we use the terminology `measurement of $M_k$' to denote a projective measurement in the eigenbasis of $M_k$, where we associate the measurement outcome $m_k(q_k)\in \zz_d$ with the eigenvalue $\omega^{m_k(q_k)}$ \footnote{Note that the projective measurement $M_k$ is only constrained on outcomes, in particular, the quantum system can be of dimension greater than $d$ for non-rank-1 projective measurements.}.

The control computer is tasked with evaluating a function $o: \zz_d^n \rightarrow \zz_d$ on some input $\textbf{i} \in \zz_d^n$. It is responsible for the side processing of classical data that determines both the measurement settings and the overall computational output from the measurement outcomes. It has very restricted computational power in that it is only capable of performing processing that is $\mathbb{Z}_d$-linear. The input $q_k$ to each party is determined in a $\mathbb{Z}_d$-linear way from the input $\textbf{i}$ only. (This type of computation is called temporally flat, in Sec.~\ref{sec: Temporal Ordering} we consider the case where the input is also determined by previous measurement outcomes.)

In particular, for each input $\textbf{i}$, the control computer evaluates $N$ linear functions $f_k: \mathbb{Z}_d^n \longrightarrow \mathbb{Z}_d$ on the input string $\mathbf{i}$, which determine the local measurement settings $q_k = f_k(\textbf{i})$ for each party. The control computer evaluates the output function $o(\mathbf{i})$ by adding the measurement outcomes $m_k(q_k)$ modulo $d$. We can say that the quantum resource state is a genuine resource, in that it increases the computational power of the control computer, if this $ld$-MBQC scheme allows for the evaluation of functions that are not $\mathbb{Z}_d$-linear. 

In order to ensure we are exposing only the computational power of the resource we require measurements to be unitarily related,
\begin{equation}\label{eqMeasSett}
M_k(q_k) =  U_k^{q_k} M_k(0) U_k^{-q_k},
\end{equation}
where the $U_k(q_k)$ form a projective representation of $\zz_d$ for some fiducial measurement choice $M_k(0)$. In doing so, we have removed the possibility of non-linear functions being introduced through the choice of measurement setting alone, thereby incidentally increasing the power of the control computer separately from the quantum resource state (cf. Appendix A in Ref.~\cite{raussendorf2016cohomological}).

In practice, we will often restrict unitaries to the Clifford group, a projective representation of the group $\mathbb{Z}_d^{2N} \rtimes \mathrm{Sp}(\mathbb{Z}_d^{2N})$ (cf. Sec.~\ref{sec: Symplectic structure of qudit stabilizer formalism}).  The fiducial measurements $M_k(0)$ are required to have a spectrum given by the $d$th roots of unity as described above.  However, we do not require the $M_k(0)$'s for different $k$ to be generalised Pauli operators in the same Pauli frame, and as such our setup is sufficient to allow for universal quantum computation.\\

In summary, we have the following definition of a temporally flat $ld$-MBQC (cf. Fig.~\ref{figldMBC}). 
\begin{definition}\label{defldMBQC}
	A temporally flat $ld$-MBQC with input string $\mathbf{i}\in \zz_d^n$ and output $o(\mathbf{i})\in \zz_d$, consists of the following components:
	\begin{enumerate}
		\item an $N$ qudit system each of local dimension $d$ where the overall resource state is represented by $\ket{\psi} \in (\mathbb{C}_d)^{\otimes N}$;
		\item a set of measurement settings $q_k = f_k(\mathbf{i})$ for some $\mathbb{Z}_d$-linear functions $f_k : \zz_d^n \rightarrow \zz_d$, independent of previous measurement outcomes;	
		\item a set of measurements $M_k$ on each qubit satisfying Eq.~(\ref{eqMeasSett}), each with $d$ possible eigenvalues $\omega^{m_k}$, where $m_k \in \zz_d$ is the measurement outcome;
		\item the computational output is a linear function of the measurement outcomes $\mathbf{m} = \{m_1,\cdots,m_N\} \in \zz_d^N$,
        \begin{equation}
        o(\mathbf{i}) = Z \mathbf{m} \mod d,
        \end{equation}
        for some $Z \in \text{Mat}_n(\zz_d)$.
	\end{enumerate}
\end{definition}
\begin{figure}[h]
	\centering
	\includegraphics[width=0.95\linewidth]{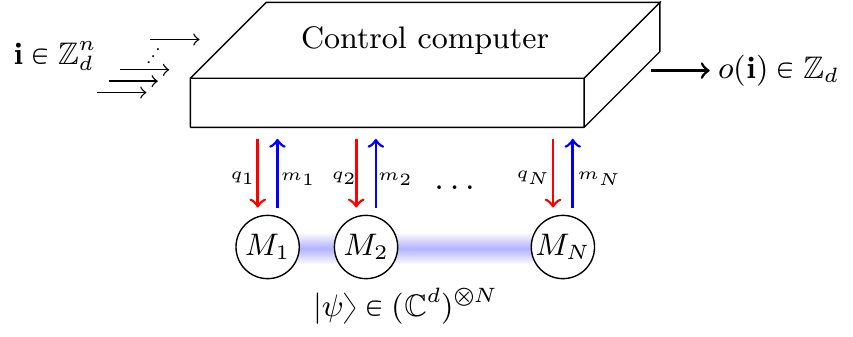}
	\caption{Schematic of the $ld$-MBQC.}
	\label{figldMBC}
\end{figure}

We remark that with a suitably chosen resource state, such as (qudit) cluster states, this model is universal for quantum computation \cite{RaussendorfBrowneBriegel2003,zhou2003quantum}.

\subsection{Structure of $ld$-MBQCs}

We first consider the case of deterministic computations for which the output function has a fixed value for every input $\mathbf{i}$, independent of the local measurement outcomes (which individually might change at different runs), and where the output is the mod-$d$ sum of the measurement outcomes. In this case, the product of all local measurements stabilizes the resource state $\ket{\psi}$, and we can describe the particular output function,
\begin{equation*}
o(\textbf{i}) = \sum_{k=1}^N m_k,
\end{equation*}
concisely through a set of eigenvalue equations,
\begin{equation}\label{eqldEig}
\bigotimes_{k=1}^N U_k^{f_k(\mathbf{i})} M_k(0) U_k^{-f_k(\mathbf{i})} |\psi\rangle = \omega^{o(\mathbf{i})} \ket{\psi},
\end{equation}
for each $\textbf{i}\in \zz_d^n$. For a given input $\textbf{i}$, we define the `global observable' $M(\textbf{i})$ as the tensor product of the local measurements. The global observables are important as their eigenvalues encode the computational output according to Eq.~(\ref{eqldEig}).  We will restrict our attention to deterministic $ld$-MBQCs throughout the paper for simplicity, however our main results generalise to the probabilistic case as discussed in Appendix~\ref{secAppC}.

\subsubsection{Non-contextual and local value assignments}

We  first  define  the  notion  of  contextuality  we consider, called strong contextuality, before  making  the  connection  with $ld$-MBQCs. We note that we only present the definition of strong contextuality as it pertains to the quantum case, for the definition that applies to general no-signalling models see Ref.~\cite{abramskybrandenburger2011}. Denote the set of observables by $\mathcal{O}$. A measurement context $\calC\subseteq\mathcal{O}$ consists of a set of mutually commuting observables, and we denote the set of all contexts by $\mathcal{M}$.  A non-contextual value assignment (NCVA) is a function $s: \mathcal{O}\rightarrow \mathbb{R}$ with the following properties
	\begin{itemize}
		\item[$(i)$] $s(A) \in \text{sp}(A)~ \forall A\in \mathcal{O}$ (that is, $s(A)$ is an eigenvalue of $A$),
		\item[$(ii)$] $\forall \mathcal{C}\in\mathcal{M}$, $s(A)s(B) = s(AB)$ $\forall A,B,AB\in \calC$.
	\end{itemize}

The NCVAs $s(A)$ should be thought of as assigning an outcome that is revealed upon measurement of $A$. Importantly, in the setting of MBQC, the value assignments must be compatible with the resource state in the following sense: we only consider NCVAs that assign a set of outcomes that can jointly occur with nonzero probability upon measurement of the resource state. If no such globally consistent NCVAs exist, then the system is strongly contextual.


In the special case of product observables $A = A_1 \otimes A_2$, $B = B_1 \otimes B_2$, $AB = A_1B_1 \otimes A_2B_2$ the locally measurable observables $A_1$ and $B_1$ ($A_2$ and $B_2$) compose individually, hence local value assignments must too (this is a special case of noncontextuality).  When no such local assignment exists, we say that the system is strongly non-local.

The connection with $ld$-MBQCs is as follows. Each input $\textbf{i}\in \zz_d^n$ can be regarded as selecting a context $\calC(\textbf{i})$ (that is, a set of commuting observables) through
\begin{equation}\label{eqContexts}
\calC(\textbf{i}) = \{M_1(q_1 = f_1(\mathbf{i})), \ldots, M_N(q_N = f_N(\mathbf{i})), M(\textbf{i})\}.
\end{equation}
We have included the global observable $M(\textbf{i}) = M_1(q_1)\otimes M_2(q_2) \otimes\cdots\otimes  M_N(q_N)$ in each context as its measurement outcome is fixed in the deterministic case, and corresponds to the computational output $o(\textbf{i})$  (that is inferred from outcomes of the local measurements). The task of finding a non-contextual hidden variable model is to find (perhaps many) value assignments to local observables that are consistent with the global value assignment. Since the global value assignment is fixed by the computational output, in general, some computations may result in an obstruction to finding a non-contextual hidden variable model in our setting.


\subsection{Anders and Browne Qubit Example}\label{sec: Anders and Browne qubit example}

Here we review the instructive example, due to Anders and Browne \cite{AndersBrowne2009}, which shows that (at least in the qubit $d=2$ case) evaluating non-linear functions deterministically with $l2$-MBQC is possible with an appropriate quantum resource state.

Following Anders and Browne, we consider an $l2$-MBQC with a three-qubit GHZ state, $\ket{\psi_{\rm GHZ}}=(|001\rangle - |110\rangle)/\sqrt{2}$, on which local measurements of Pauli observables $X$ or $Y$ on each qubit enable the deterministic computation of the non-linear NAND gate. The control computer receives two bits $\textbf{i}=(i_1,i_2) \in \zz_2^2$ as input. The classical pre-processing to determine the measurement settings on each qubit amounts to evaluating the linear functions $f_1(\mathbf{i}) = i_1$, $f_2(\mathbf{i}) = i_2$ and $f_3(\mathbf{i}) = i_1 \oplus i_2$. The bits $q_k = f_k(\textbf{i})$ determine the measurement setting on each qubit according to $M_k(0) = X$ and $M_k(1) = Y$, for $k \in \{1,2,3\}$. These measurement settings are related by the unitary transformation $U = \frac{1}{\sqrt{2}}(X + Y)$. If the eigenvalue $+1$ (${-}1$) is observed, then the outcome $m_k(q_k)=0$ ($m_k(q_k)=1$) is recorded. 

These measurement settings define the global observables,
\begin{align}
M(0,0) &= X \otimes X \otimes X, \\
M(1,0) &= Y \otimes X \otimes Y, \\
M(0,1) &= X \otimes Y \otimes Y, \\
M(1,1) &= Y \otimes Y \otimes X,
\end{align}
with the state $\ket{\psi_{\rm GHZ}}$ an eigenvector of each observable, and with corresponding eigenvalues given by 
\begin{equation}\label{eqNAND}
({-}1)^{o(i_1,i_2)} = ({-}1)^{i_1 i_2 + 1} = ({-}1)^{\text{NAND}(i_1,i_2)}.
\end{equation}

In other words, the measurement outcomes of the local $X$ and $Y$ measurements can be linearly processed to compute the output function $o(\mathbf{i}) = \sum_{k=1}^3 m_k(\mathbf{i})$, which by Eq.~(\ref{eqNAND}) gives $o(i_1,i_2) = \text{NAND}(i_1,i_2)$.  Thus, the $l2$-MBQC evaluates a non-linear function of the input that could not be evaluated by the control computer alone.


\subsection{Central Questions}\label{sec: Central Questions}

What $ld$-MBQCs can compute non-linear functions?  And what properties of the quantum resource state enable this additional power?  In the qubit case, the conditions under which an $l2$-MBQC allows for the computation of non-linear boolean functions---functions that would otherwise be beyond the capabilities of the control computer---have been well characterised.  In particular, Raussendorf has shown that any $l2$-MBC and thus any $l2$-MBQC that computes a non-linear boolean function constitutes a proof of strong contextuality~\cite{Raussendorf2009}; if a $l2$-MBQC can be described by a non-contextual `hidden variable model', where the outcomes associated with measurements are pre-determined by (local) value assignments, it is restricted to computing linear functions. (Note that this result also holds in the temporally ordered case, where measurement settings can be additionally determined by past measurement outcomes.)  The above example is strongly contextual, and so there does not exist any assignment of pre-determined measurement outcomes to each of the local observables that can reproduce the correlations required for the computation of the NAND gate. We restate Raussendorf's theorem of Ref.~\cite{Raussendorf2009}. 

\begin{theorem*}[Raussendorf]\label{thm: Raussendorf qubit contextuality}
	Be $M$ an $l2-MBC$ which deterministically evaluates a boolean function $o: \mathbb{Z}_2^n \longrightarrow \mathbb{Z}_2$. If $o(\mathbf{i}) \in \mathbb{Z}_2$ is non-linear $\mathrm{mod} \ 2$ in $\mathbf{i} \in \mathbb{Z}_2^n$ then $M$ is strongly contextual.
\end{theorem*}

The qudit case with $d\geq 3$ is much less explored. While qubits are natural to consider in (measurement-based) quantum computation, they are pathological from the perspective of contextuality. Single qubits are non-contextual by the Kochen-Specker theorem, while multiple qubits exhibit state-independent contextuality using only Pauli observables in contrast to its qudit counterparts. A natural question is  whether the interplay between contextuality and non-linearity holds more generally in $ld$-MBQCs with $d\geq 3$.

In the following section, we will show that the qudit case is not so straightforward, and certain kinds of non-linear functions may be computed even with non-contexual $ld$-MBQCs.


\section{Examples and Puzzles}\label{sec: Examples and puzzles}

In this section, we illustrate some of the subtleties involved in the qudit case. We focus on a particularly interesting class of $ld$-MBQCs based on the qudit stabilizer formalism. Unlike the qubit case, such $ld$-MBQCs are non-contextual.  (An explicit non-contextual hidden variable model for the qubit stabilizer theory is given by the discrete Wigner function~\cite{Gross2006,veitch2012negative}. There has been a considerable amount of recent research investigating the differences between the qubit and qudit stabilizer subtheories from the perspective of contextuality; see for example~\cite{Gross2006,abramsky2017complete,CataniBrowne2017,lillystone2018contextuality,karanjai2018contextuality,raussendorf2017contextuality,delfosse2015wigner,bermejo2017contextuality}.)  

Contrary to what one might naively expect, we will see that qudit stabilizer $ld$-MBQCs possess a computational power that exceeds $\mathbb{Z}_d$-linear processing.  That is, non-linear functions can be evaluated using an MBQC that is entirely non-contextual, in stark contrast to the qubit case. This demonstrates that the relationship between strong contextuality and non-linearity in the qubit case is not the end of the story, and for qudits with $d \geq 3$ we need a finer functional constraint.

In this section we will be restricting to the case where the measurements $M_k$ belong to the qudit Pauli group, and the unitaries $U_k$ that relate these measurements as in Eq.~(\ref{eqMeasSett}) belong to the Clifford group.  Important to our considerations is that when $d$ is odd, this subtheory admits a non-contextual description~\cite{Gross2006}. We outline a very restricted case under which linearity in the output can be recovered, along with two examples within this non-contextual framework that result in non-linear output functions.


\subsection{Symplectic Structure of Qudit Stabilizer Formalism}\label{sec: Symplectic structure of qudit stabilizer formalism}

Our results will make extensive use of the symplectic structure of the qudit Clifford group, and so we briefly review this formalism here.  For more details, see Ref.~\cite{deBeaudrap2013}.

Recall that the Pauli group $\mathcal{P}^{\otimes N}_d$ over $\mathbb{Z}_d$ is the group generated by $N$-fold tensor products of individual elements from $\langle X_k,Z_k,\omega \mathbbm{1}_k \rangle$, $k \in \{1,\cdots,N\}$ with $X |z\rangle = |z+1\rangle$, $Z |z\rangle = \omega^z |z\rangle$, and $\omega = e^{\frac{2\pi i}{d}}$ a $d$-th root of unity. 

These qudit Pauli operators can be conveniently represented (up to phase) by Weyl operators. A Weyl operator $W_\mathbf{v}$ for $\mathbf{v} = (\mathbf{a},\mathbf{b})^\intercal \in \mathbb{Z}^{2N}_d$ is defined as the $N$-fold tensor product of local operators $W_{a,b} = \tau^{-ab}Z^aX^b$ where $\tau^2 = \omega$ and $a,b \in \mathbb{Z}_d$.  Note that Weyl operators are generalised Pauli operators with a particular choice of phase.

The Clifford group $\mathcal{C}_N(d) \subset \mathcal{U}(\mathcal{H}^{\otimes N}_d)$ of $\mathcal{P}^{\otimes N}_d$ is the group of unitary operators such that $VPV^\dagger \in \mathcal{P}^{\otimes N}_d$ for all $P \in \mathcal{P}^{\otimes N}_d$, $V \in \mathcal{C}_N(d)$.  All ($N$-qudit) Clifford operators $V \in \mathcal{C}_N(d)$ factorise,
\begin{equation*}
V = U W_\mathbf{x}, \quad \mathbf{x} \in \mathbb{Z}^{2N}_d,
\end{equation*}
into a Weyl operator $W_\mathbf{x}$ and an element of the group of symplectic Clifford operators $U \in \sigma\mathcal{C}_N(d)$. These are defined as automorphisms on the set of Weyl operators, i.e. for all $\mathbf{v} \in \mathbb{Z}_d^{2N}$ it holds that $U W_\mathbf{v} U^\dagger = W_\mathbf{w}$ for some $\mathbf{w} \in \mathbb{Z}_d^{2N}$, in fact, they preserve the underlying symplectic structure,
\begin{equation*}\label{eq: symplectic Clifford group}
U W_\mathbf{v}U^{-1} = W_{C_U\mathbf{v}}, \ \mathrm{for} \ \mathrm{some} \ C_{U} \in \mathrm{Sp}_{2N}(\mathbb{Z}_d).
\end{equation*}
The group $\mathrm{Sp}_{2N}(\mathbb{Z}_d)$ denotes the group of symplectic transformations, i.e. of linear transformations $C: \mathbb{Z}^{2N}_d \longrightarrow \mathbb{Z}^{2N}_d$ such that $C^T \sigma_{2N} C = \sigma_{2N}$ where the symplectic matrix is given as $\sigma_{2N} = \begin{bmatrix}
0_N & \mathbbm{1}_N \\
-\mathbbm{1}_N & 0_N
\end{bmatrix}$. Moreover, the symplectic inner product is given by $\left[\mathbf{v},\mathbf{w}\right] = \textbf{v}^T \sigma_{2n} \textbf{w}$, for $\textbf{v}, \textbf{w} \in \zz_d^{2N}$.

\subsection{Computational Output and (Non-)Linearity}
Let us now relate the transformation properties of Pauli observables under Clifford operations to the computational output of the $ld$-MBQC, using the symplectic formalism.  From the defining commutation relation of Weyl operators, 
\begin{equation}\label{eq: Weyl commutator}
W_{\mathbf{v}}W_\mathbf{w} = \omega^{\left[\mathbf{v},\mathbf{w}\right]}W_\mathbf{w} W_\mathbf{v} \, ,
\end{equation}
and the fact that symplectic operators preserve the symplectic inner product, we obtain the following relation for individual qudits (for clarity we omit the subscript $k$ labeling different qudit sites),
\vspace{0.1cm}

\begin{align}\label{eq: general functional dependence}
&V^{f(\mathbf{i})} \ W_\mathbf{v} \ V^{-f(\mathbf{i})} \notag \\ 
& = (U W_\mathbf{x})^{f(\mathbf{i})} \ W_\mathbf{v} \ (U W_\mathbf{x})^{-f(\mathbf{i})} \notag \\
& = (UW_\mathbf{x})^{f(\mathbf{i}) - 1} \ W_{C_U\mathbf{x}} \  W_{C_U\mathbf{v}} \ W_{-C_U\mathbf{x}} \ (W_{-\mathbf{x}}U^{-1})^{f(\mathbf{i})-1} \notag \\
& = (W_{C_{U}\mathbf{x}} \ \cdots \ W_{C_{U}^{f(\mathbf{i})}\mathbf{x}}) \  W_{C_{U}^{f(\mathbf{i})}\mathbf{v}} \  (W_{-C_{U}^{f(\mathbf{i})}\mathbf{x}} \cdots \ W_{-C_{U}\mathbf{x}}) \notag \\
& = \omega^{\left[ C_U^{} \mathbf{x},C_U^{f(\mathbf{i})}\mathbf{v} \right] + \left[ C_U^2\mathbf{x},C_U^{f(\mathbf{i})}\mathbf{v} \right] + \ \cdots \ + \left[ C_U^{f(\mathbf{i})}\mathbf{x},C_U^{f(\mathbf{i})}\mathbf{v} \right]} W_{C_U^{f(\mathbf{i})}\mathbf{v}} \notag \\
& = \omega^{ \sum_{k=0}^{f(\mathbf{i})-1} [\mathbf{x},C^k_{U}\mathbf{v}]} W_{C_{U}^{f(\mathbf{i})}\mathbf{v}},
\end{align}

for any Clifford unitary $V \in \mathcal{C}_N(d)$. Note that the phase in Eq.~(\ref{eq: general functional dependence}) is state-independent, it only depends on the Weyl commutation relations.


\subsubsection{Example 1: Linear Output}\label{sec: Restricting to linear functions}

As a first example, we examine the very restrictive case, where the controlled unitary operators in Def.~\ref{defldMBQC} are Pauli operators. Using only controlled Pauli operators in Eq.~(\ref{eq: general functional dependence}), i.e., $V = W_\mathbf{x}$, $\mathbbm{1} = U \in \sigma\mathcal{C}_1(d)$, the phase depends linearly on the input functions $f(\mathbf{i})$. In fact, we simply obtain a variant of Eq.~(\ref{eq: Weyl commutator}),
\begin{equation}\label{eq: Weyl produces linear functions}
W_\mathbf{x}^{f(\mathbf{i})} W_\mathbf{v} W_\mathbf{x}^{-f(\mathbf{i})} = \omega^{f(\mathbf{i})\left[\mathbf{x},\mathbf{v}\right]} W_\mathbf{v} \, .
\end{equation}
From  Eq.~(\ref{eq: Weyl produces linear functions}), we infer that conjugation of a Pauli operator by Pauli operators results in multiplication of a phase, yet does not change the context. That is $M(\textbf{i}) \propto M(\textbf{0})$ and $\calC(\textbf{i}) \propto \calC(\textbf{0})$ for all inputs, meaning the output $o(\textbf{i})$ is linearly related to $o(\textbf{0})$. As a result, we are trivially restricted to be non-contextual.


\subsubsection{Example 2:  Quadratic Output}\label{sec: Example 1:  Quadratic functions}

In the next two examples we remain within the stabilizer subtheory but extend to arbitrary Clifford unitaries. First, we show how using control unitaries that are symplectic  Cliffords (i.e., Cliffords that are not Paulis) allows us to compute quadratic functions from the symplectic inner product. In the following, we assume $d\geq 3$.

The generalised phase gate $S$ is an element of the symplectic Clifford group,
\begin{equation*}
S = \sum_{z=0}^{d-1} \tau^{z^2} |z\rangle \langle z| \ \in\ \sigma \mathcal{C}_1(d) \, ,
\end{equation*}
which up to phase acts on the generalised Pauli $X = W_\mathbf{v}, \mathbf{v} = (0,1)^\intercal$ by multiplication with Pauli $Z$, $S^kXS^{-k} = \tau^{-k}Z^kX$. Consider the following state of $N=2d$ qudits,
\begin{equation*}
|\psi\rangle = \frac{1}{\sqrt{d}}\sum_{z=0}^{d-1} \ket{z}^{\otimes 2} \ket{z + 1}^{\otimes 2} \cdots \ket{ z+d-1}^{\otimes 2} .
\end{equation*}

We fix all of the linear functions $f_k(\textbf{i}) = f(\textbf{i})$ to be the same, and note that we have the following stabilizer relations,
\begin{equation}\label{eq: stabilizer relation}
\bigotimes_{k=1}^{2d} \left( S^{f(\mathbf{i})}W_\mathbf{v}S^{-f(\mathbf{i})} \right)_k |\psi\rangle = |\psi\rangle, \quad \mathbf{v} = (0,1)^\intercal .
\end{equation}
where the parentheses $( \cdot )_k$ denote the subsystem on which the operator acts.

We can use this setup together with the symplectic structure of Weyl operators to implement quadratic output functions through accumulated symplectic products.  Without loss of generality, we choose the first qudit and take $V_1 = (SW_\mathbf{x})_1$ for $\mathbf{x} = (0,-1)^\intercal$ in Eq.~(\ref{eq: stabilizer relation}) such that $[\mathbf{x},C_S^k\mathbf{v}] = k$, while leaving $V_k = S_k$ for $k\geq2$. Hence,
\begin{equation*}
o(\mathbf{i}) = \sum_{k=0}^{f(\mathbf{i})-1} [\mathbf{x},C_S^k\mathbf{v}]
= \frac{f(\mathbf{i})(f(\mathbf{i})-1)}{2}.
\end{equation*}
Despite $C_S$ being a linear map, it leads to quadratic output functions due to the symplectic structure of the Weyl group.


\subsubsection{Example 3: General Non-Linear Output}\label{sec: Example 3:  more non-linear functions}

With our final example, we show that one can even go beyond quadratic functions. Another symplectic Clifford operator is given by,
\begin{equation}
\label{eq:Mr}
M_u := \sum_{k=0}^{d-1} |uk\rangle \langle k|, \quad u \in \mathbb{Z}_d^\times ,
\end{equation}
where $R^\times$ denotes the multiplicative group of units within a ring $R$.

Now let $|\psi \rangle = |1\rangle \in \mathbb{C}_d$, $\mathbf{v} = (1,0)^\intercal \in \mathbb{Z}_d^2$ and consider the symplectic Clifford unitary $U = M_u^{f(\mathbf{i})} \in \sigma\mathcal{C}_1(d)$ for some $u \in \mathbb{Z}_d^\times$ acting on $Z = W_\mathbf{v}$,
\begin{align}\label{eq: construct exponential}
U W_{\mathbf{v}}U^{-1}|\psi\rangle &= M_u^{f(\mathbf{i})}W_{\mathbf{v}}M_u^{-f(\mathbf{i})}|\psi\rangle \notag \\
&= W_{C_{M_u}^{f(\mathbf{i})}\mathbf{v}}|\psi\rangle \notag \\
&= \omega^{u^{-f(\mathbf{i})}}|\psi\rangle .
\end{align}

The output function $o(\mathbf{i}) = u^{-f(\mathbf{i})}$ is again non-linear, however, the underlying system is part of the qudit stabilizer formalism and can be given a non-contextual hidden variable model with local value assignments $m_k(\mathbf{i})$ such that $o(\mathbf{i}) = \sum_{k=1}^N m_k(\mathbf{i})$ \cite{deBeaudrap2013, Gross2006}.\medskip

Examples 2 and 3 are therefore in clear contrast to Raussendorf's result in the qubit case, as they allow for the evaluation of non-linear functions with MBQCs that are non-contextual. It turns out that linearity is a special property of boolean functions which does not translate to the qudit case.


\section{Finite Fields and Local Universality}\label{secLocalUniv}

Given these conceptual differences between the qubit and qudit cases as illustrated by the examples in the previous section, it will be helpful to take a more abstract view on the problem.  We begin by reviewing some properties of functions on finite fields.


\subsection{Functions on Finite Fields}\label{sec: Functions on Finite Fields}

Let $\mathbb{F}_d$ be the finite field with $d = p^r$ ($p$ prime and $r \in \mathbb{N}$) elements and $\Omega^{\mathbb{F}_d}_n := \mathbb{F}_d[x_1,\ldots,x_n]$ the polynomial ring in $n$ variables, $x_1,\ldots,x_n \in \mathbb{F}_d$.

For infinite fields there are many non-polynomial functions, but not so for finite fields.  
\vspace{0.1cm}

\begin{theorem*}[Functions on finite fields]\label{thm: polynomials in finite fields}
Let $\mathbb{F}_d$ be the finite field with $d$ elements, and $n \in \mathbb{N}$. Then every function $g : \mathbb{F}_d^n \longrightarrow \mathbb{F}_d$ is given by a polynomial, $g \in \Omega^{\mathbb{F}_d}_n$ of partial degree  less or equal to $d-1$ in each variable $x_i$.
\end{theorem*}

\begin{proof}
See Appendix~\ref{secAppA}.
\end{proof}

Here, the \textit{partial degree} of a variable within a product is the exponent of that variable.\medskip

For qubits, Raussendorf proved that computability of a non-linear function using an l2-MBC implies strong contextuality of the underlying system~\cite{Raussendorf2009}.  We note that this result depends critically on the application of the above theorem to $d=2$, where it states that all functions $g: \mathbb{Z}_2 \rightarrow \mathbb{Z}_2$ are linear.  In a non-contextual hidden variable model, the measurement outcomes $m_k$ are determined by a local value assignment that can depend on the choice of measurement setting $q_k = f_k(\mathbf{i})$; these local value assignments $m_k(\mathbf{i})$ are necessarily linear by the above theorem.  In the same way, the final output function $o(\mathbf{i})$ on an $l2$-MBQC on $N$ qubits is a linear combination of the measurement outcomes at each site,
\begin{equation}\label{eqOutputRelation}
o(\mathbf{i}) = \sum_{k=1}^N m_k(\mathbf{i})\, .
\end{equation}
If $o(\mathbf{i})$ is non-linear, then the assumption of non-contextuality must be incorrect.  Non-linearity in this case can only arise from products of at least two different inputs, as in Anders and Browne's NAND gate illustrated in Sec.~\ref{sec: Anders and Browne qubit example}.  

In generalising to the qudit case, however, a local value assignment is not required to be a linear function of the choice of measurement.  Here, this simple connection between non-contextuality and linearity is lost.  Nevertheless, we now show how to use the above theorem for functions on finite fields to build a new, general connection.


\subsection{Linearly Closed Subspaces}\label{sec: Linearly Closed Subspaces}

We would like to characterise the space of all functions of the form $g: \mathbb{F}_d^n \longrightarrow \mathbb{F}_d$ in terms of non-trivial subspaces under linear pre- and post-processing, which we denote by $\subset_l$. First, define the space of linear functions as 
	\begin{equation*}\label{eq: linear subspace}
	L^{\mathbb{F}_d}_n := \{l \in \Omega^{\mathbb{F}_d}_n \mid l(x) = c_0 + \sum_{j=1}^n c_jx_j, \ x \in \mathbb{F}_d^n, c_j \in \mathbb{F}_d \},
	\end{equation*}
where $x = (x_1, \ldots , x_n)$. Then a schematic that represents the linear pre- and post-processing that we consider is given in Fig.~\ref{fig: schematic}.  Every measurement outcome $m_k: \mathbb{F}^n_d \rightarrow \mathbb{F}_d$ is a function from the inputs $\mathbf{i} \in \mathbb{F}^n_d$ to measurement outcomes that decomposes into a linear function $f_k \in L^{\mathbb{F}_d}_n$ and a contribution $\phi_k: \mathbb{F}_d \rightarrow \mathbb{F}_d$,
\begin{equation}
m_k = \phi_k \circ f_k, \quad \forall k \in \{1,\cdots,N\}.
\end{equation}
We refer to $\phi_k$s as correlation functions.

\captionsetup[subfigure]{position=top, labelfont=bf,textfont=normalfont,singlelinecheck=off,justification=raggedright}
\begin{figure}[h]
	\centering
	\subfloat[]{	\includegraphics[width=0.95\linewidth]{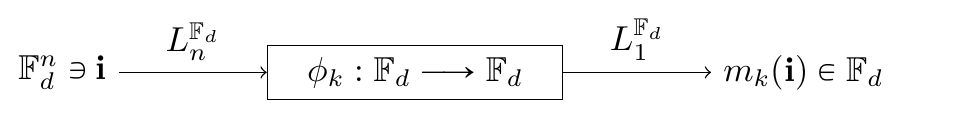}}	\label{figLinearA}
	\quad
	\subfloat[]{	\includegraphics[width=0.95\linewidth]{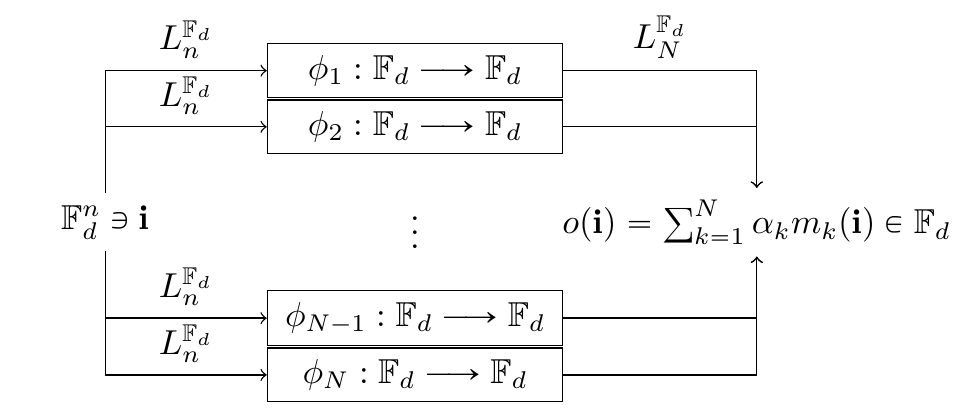}}	\label{figLinearB}
	\caption{Schematic of functional signatures within the general correlation setup introduced by Anders and Browne~\cite{AndersBrowne2009}. (a) In a non-contextual setting, local value assignments $m_k: \mathbb{F}^n_d \rightarrow \mathbb{F}_d$ split into (classical) linear pre- and post-processing and local (quantum) measurements $\phi_k$. (b) The same holds for the output function $o: \mathbb{F}_d^n \rightarrow \mathbb{F}_d$, $o(\textbf{i}) = \sum_{k=1}^N \alpha_k m_k(\textbf{i})$ for some $\alpha_k \in \mathbb{F}_d$. Any additional complexity arises from the (quantum) measurements $\phi_k: \mathbb{F}_d \rightarrow \mathbb{F}_d$.}
	\label{fig: schematic}
\end{figure}

Note first that the space of functions available after pre- and post-processing will at least contain the original function space, as the identity is a linear function. Consider the two trivial cases: the space of all functions $\Omega^{\mathbb{F}_d}_n$; and the space of linear functions $L^{\mathbb{F}_d}_n$. Any linear combination of linear functions results again in a linear function, hence, we find $L^{\mathbb{F}_d}_n \subset_l \Omega^{\mathbb{F}_d}_n$.

Aside from these two trivial cases, there also exist non-trivial spaces that are stable under linear pre- and post-processing.  Define the following subspaces for $1 \leq \delta \leq n(d-1)$, using the notation $\langle \cdot \rangle_l$ for linear span,
\begin{equation}\label{defn: linear subspaces}
\Omega^{\mathbb{F}_d}_n(\delta) := \Big\langle \prod_{j=1}^n x_j^{a_j} \ \Big\mid \ a_j \in \mathbb{F}_d, \ \sum_{j=1}^n a_j \leq \delta \ \Bigr\rangle_l.
\end{equation}
The function spaces $\Omega^{\mathbb{F}_d}_n(\delta)$ depend on the field $\mathbb{F}_d$, the number of inputs $n$ and the maximal combined degree of monomials $\delta$. (The \textit{combined degree} within a product of variables is the sum of their respective exponents as in Eq.~(\ref{defn: linear subspaces}).)  In other words, $\Omega^{\mathbb{F}_d}_n(\delta)$ contains all polynomials $g: \mathbb{F}_d^n \longrightarrow \mathbb{F}_d$ of degree at most $\delta$.

We prove a lemma, detailing the behaviour of these subspaces under linear pre- and post-processing.

\begin{lemma}\label{lm: linearly stable subspaces}
Let $g \in \Omega^{\mathbb{F}_d}_n$  be a polynomial of degree $|g| = \delta \leq n(d-1)$, then the set of all polynomials generated by linear pre- and post-processing coincides with $\Omega^{\mathbb{F}_d}_n(\delta)$, i.e.
\begin{equation*}
L^{\mathbb{F}_d}_1 \circ g \circ L^{\mathbb{F}_d}_n = \Omega^{\mathbb{F}_d}_n(\delta) \, .
\end{equation*}
\end{lemma}

\begin{proof}
We focus on the space of functions generated by linear pre-processing first. We find that $g \circ L^{\mathbb{F}_d}_n \subseteq_l \Omega^{\mathbb{F}_d}_n(\delta)$ as evaluating monomials $x^j$ on linear functions results in polynomials of degree at most $j$. As the definition of Eq.~(\ref{defn: linear subspaces}) already captures linear post-processing, we also have
\begin{equation*}
L^{\mathbb{F}_d}_1 \circ g \circ L^{\mathbb{F}_d}_n \subseteq_l L^{\mathbb{F}_d}_1 \circ \Omega^{\mathbb{F}_d}_n(\delta) = \Omega^{\mathbb{F}_d}_n(\delta)\,.
\end{equation*}

On the other hand we can always choose a linear function such that after evaluation on $x^j$ the resulting polynomial and $x^j$ are linearly independent (by producing other terms of less degree). Hence, taking linear combinations we can generate all polynomials of degree at most $j$,
\begin{equation*}
L^{\mathbb{F}_d}_1 \circ g \circ L^{\mathbb{F}_d}_n \supseteq_l \Omega^{\mathbb{F}_d}_n(\delta) \, .
\end{equation*}
\end{proof}

We conclude that the subspaces closed under linear pre- and post-processing are exactly $\Omega^{\mathbb{F}_d}_n(\delta)$ for $1 \leq \delta \leq n(d-1)$, where $\Omega^{\mathbb{F}_d}_n(1) = L^{\mathbb{F}_d}_n$.


\subsection{Local Universality}\label{sec: Local Universality}

In the qubit case, $\mathbb{F}_2 = \mathbb{Z}_2$, every local output function $m_k: \mathbb{Z}_2^n \longrightarrow \mathbb{Z}_2$ is linear. Put another way, with any given function $\phi \in \Omega_1^{\mathbb{F}_d}$ together with linear pre- and post-processing, we can realise arbitrary functions $m_k: \mathbb{Z}_2^n \longrightarrow \mathbb{Z}_2$.  We say (somewhat trivially) that linear functions are \textit{universal} for local qubit value assignments.

For qudits ($d = p^r$, $p$ prime, $r \in \mathbb{N}$ and corresponding finite field $\mathbb{F}_d$), local measurement outcomes are still maps of the form, $\phi_k: \mathbb{F}_d \longrightarrow \mathbb{F}_d$.  However, it is no longer true that all functions of this type are linear; any monomial of degree higher than one is clearly not.

Hence, whereas linear functions turned out to be universal (in the sense of being able to evaluate \textit{any} function $\phi_k: \mathbb{Z}_2 \longrightarrow \mathbb{Z}_2$), this is not the case for higher dimensional systems where there can be more than two measurement outcomes.  Therefore, we define a new property capturing this universality for both cases.

\vspace{0.1cm}

\begin{definition}[Local universality]\label{defn: local universality}
A $ld$-MBC is called locally universal if it implements all functions  $\phi:\mathbb{F}_d \longrightarrow \mathbb{F}_d \in \Omega^{\mathbb{F}_d}_1$.
\end{definition}

Note that the system size crucially enters this definition by means of the finite field $\mathbb{F}_d$, where $d$ denotes the qudit dimension.

With this definition, we are now able to identify the generalised connection between computational power and contextuality for qudit systems.  In a non-contextual hidden variable model, we again have local output functions $m_k: \mathbb{F}_d^n \longrightarrow \mathbb{F}_d$; these are not required to be linear in the qudit case.  Nonetheless, as we now show, this does not allow for the evaluation of any non-linear function.  For qubits, non-linear functions arose (by necessity) from crossterms, i.e., terms with combined degree greater than 1, such as the term $i_1i_2$ in the NAND gate in the example of Sec.~\ref{sec: Anders and Browne qubit example}.  For qudits, local universality allows us to implement some crossterms for finite fields with more than two elements; however, by Lem.~\ref{lm: linearly stable subspaces} we find that we can only implement a strict subset within the space of all polynomials,
\begin{equation*}
\Omega^{\mathbb{F}_d}_n(\delta) \ \subsetneq \ \Omega^{\mathbb{F}_d}_n \quad \forall \delta \in \{1,\ldots,d-1\} \, .
\end{equation*}
Local universality is a restriction on the space of (global) functions $o: \mathbb{F}_d^n \longrightarrow \mathbb{F}_d$, but a necessary requirement to maximise the computational power of correlations in $\phi_k: \mathbb{F}_d \longrightarrow \mathbb{F}_d$. In short, locally universal models are in the class $\Omega^{\mathbb{F}_d}_n(d-1)$ and,
\begin{equation*}
\Omega^{\mathbb{F}_d}_n(\delta) \ \subseteq \ \Omega^{\mathbb{F}_d}_n(d-1) \ \subsetneq \ \Omega^{\mathbb{F}_d}_n, \quad \delta \in \{1,\ldots,d-1\} .
\end{equation*}

The key observation is that for both qubits and qudits non-contextual models are at most locally universal, slightly more general we have:

\begin{theorem}\label{thm: generalised correspondence contextuality and computational complexity}
Consider a deterministic $ld$-MBC of prime power dimension $d = p^r$, $p$ prime, $r \in \mathbb{N}$. If the resource state is not strongly non-local then the output polynomial has maximal degree bounded by $d$, i.e.
\begin{equation}
o(\mathbf{i}) \in \Omega^{\mathbb{F}_d}_n(d-1).
\end{equation}
\end{theorem}

\begin{proof}
The theorem is a direct consequence of Lem.~\ref{lm: linearly stable subspaces} and the fact that the correlation functions $\phi_k$ are restricted to the class $\Omega_1^{\mathbb{F}_d}(\delta)$. Here we use strong non-locality in the sense of \cite{abramskybrandenburger2011}.
\end{proof}

Note that Thm.~\ref{thm: generalised correspondence contextuality and computational complexity} is independent of the particular physical implementation of these function evaluations. In the next section we will look at the particular case of $ld$-MBQC.\\


\section{Non-Locality in MBQC}\label{sec: Non-locality in MBQC}

In this section we connect these results with $ld$-MBQC as outlined in Def.~\ref{defldMBQC}. We have the immediate corollary of Thm.~\ref{thm: generalised correspondence contextuality and computational complexity}.

\begin{corollary}\label{cor: generalised correspondence contextuality and computational complexity}
Let $M$ be a $ld$-MBQC for $d$ prime, which deterministically evaluates a function $o: \mathbb{Z}_d^n \longrightarrow \mathbb{Z}_d$. If $o(\mathbf{i})$ (when written as a polynomial) involves at least one term of the form $\prod_{j=1}^n \ x_j^{a_j}$, s.t. $\sum_{j=1}^n a_j \geq d$, then $M$ is strongly contextual, and specifically strongly non-local.
\end{corollary}

We note that Cor.~\ref{cor: generalised correspondence contextuality and computational complexity} was first reported in Ref.~\cite{HobanWallmanBrowne2011}. In order to see how tight this bound is, we check whether general $ld$-MBQC instances, and specifically non-contextual ones, satisfy local universality. For qubits, this is trivially the case. In this section, we prove that local universality rather than linearity also generalises to $ld$-MBQCs on qudits of prime dimension.

We give a brief note regarding our restriction to prime values of $d$.  Recall that in the setting of $ld$-MBQC, we are dealing with rings of integers modulo $d$. For $d$ a prime number, $\mathbb{Z}_d$ is in fact a field, but not all finite fields arise this way. Thm.~\ref{thm: generalised correspondence contextuality and computational complexity} holds for all finite fields $\mathbb{F}_d$, i.e., also fields of order a prime power $d = p^r$.  However, $\mathbb{F}_d \not\cong \mathbb{Z}_d$ whenever $r \geq 2$, as the latter is a ring but not a field. Functions on unital, commutative rings are not polynomials in general.  For these reasons, we will restrict to fields $\mathbb{Z}_p$ in this section, but we note that our results allow for partial inferences in the case of rings $\mathbb{Z}_d$ with $d$ odd as well.


\subsection{Local Universality in MBQC}\label{sec: Local Universality in MBQC}

We show that functions arising from $ld$-MBQCs with single-qudit Clifford unitaries as the classical control operations are locally universal. In fact, this is true even under the restriction to stabilizer states.

\begin{theorem}\label{thm: local universality in MBQC}
The stabilizer subtheory in $ld$-MBQC with $d$ prime is locally universal.
\end{theorem}

\begin{proof}
See Appendix~\ref{secAppB}.
\end{proof}

Note that for general rings $\mathbb{Z}_d$ and $d$ odd, the function space $\tilde{\Omega}_n^{\mathbb{Z}_d} ∶= \{g:\mathbb{Z}^n_d \rightarrow \mathbb{Z}_d\} \supsetneq \mathbb{Z}_d[x_1,⋯,x_n]$ contains non-polynomial functions. Nevertheless, the proof of Thm.~\ref{thm: local universality in MBQC} extends to rings $\mathbb{Z}_d$ for $d$ odd, hence, Cor.~\ref{cor: generalised correspondence contextuality and computational complexity} remains true in those cases whenever the output function is a polynomial.\\

In Sec.~\ref{sec: Restricting to linear functions}, we proved that using Pauli unitaries for classical control is not powerful enough to yield non-linear functions, yet within the full stabilizer subtheory of Clifford unitaries the MBQCs allow for non-linear output functions. Thm.~\ref{thm: local universality in MBQC} shows that Clifford unitaries already generate all non-linear functions, $\phi: \mathbb{Z}_d \longrightarrow \mathbb{Z}_d$, in particular, the bound in Cor.~\ref{cor: generalised correspondence contextuality and computational complexity} is tight.

Note also that the function space accessible using the stabilizer subtheory contains at least $\Omega_n^{\mathbb{Z}_d}(d-1)$ by Thm.~\ref{thm: local universality in MBQC}. On the other hand the existence of a non-negative discrete Wigner function restricts the stabilizer subtheory to local universality.

\begin{corollary}\label{cor: stabilizer subtheory}
The stabilizer subtheory in $ld$-MBQC with $d$ prime is strictly bounded by local universality, i.e. its computation class is $\Omega_n^{\mathbb{Z}_d}(d-1)$.
\end{corollary}

\begin{proof}
The discrete Wigner function provides a non-contextual description for any implementation of the stabilizer subtheory~\cite{Gross2006}. Hence, such implementations are not strongly contextual. Yet any implementation evaluating a function $o(\mathbf{i}) \notin \Omega_n^{\mathbb{Z}_d}(d-1)$ is strongly contextual by Cor.~\ref{cor: generalised correspondence contextuality and computational complexity}.
\end{proof}

In particular, we cannot harness any computational power from non-local correlations using only stabilizer states.\\

Up until now we have focussed on deterministic MBQCs.  We note however, in analogy to the qubit case, Cor.~\ref{cor: generalised correspondence contextuality and computational complexity} extends to a probabilistic version for success probabilities within a finite interval around $p_S = 1$.  We provide the details of this generalisation in Appendix~\ref{secAppC}.


\subsection{Scaling under Composition}\label{sec: Scaling under Composition}

At the core of the framework of $ld$-MBC is the identification of locally-measureable systems, and the power of correlations between these systems. Within this framework, we make a distinction between contextuality `within' a local system and the non-locality between them. So far we have shown how non-locality leads to correlations that can be used as a resource for computational power. To further illustrate this connection, we consider (somewhat reversely) how computational power constrains the composition of systems. To this end we assume a locally universal $ld$-MBC with finite field $\mathbb{F}_d$, rather than restricting to $\mathbb{Z}_d$ in $ld$-MBQC.

From the results on linearly invariant subspaces in Sec.~\ref{secLocalUniv}, we know that any restriction on functional computability lies in the maximal degree of the function's monomials. Under the constraint of locality, functions can be computed locally only and combined after, leading to $\Omega_n^{\mathbb{F}_d}(d-1)$ in Thm.~\ref{thm: generalised correspondence contextuality and computational complexity}. However, when taken as a single system the individual systems can be combined first and functions computed after. Now assume we are given a device that allows us to perform non-local (correlated) measurements and thus to combine systems first before computing. It is clear that with such a device we can again implement arbitrary functions $g: \mathbb{F}_d^n \longrightarrow \mathbb{F}_d$. Given these preliminary considerations, we consider the two scenarios in which dimensions of subsystems, i.e. the number of different local outcomes (cf. $\mathbb{Z}_l$ in Sec.~\ref{sec: The Setup}), were to combine additively vs multiplicatively under composition.

Take a composite system of dimension $d=p^r$, $p$ prime, $r \in \mathbb{N}$. Any decomposition into subsystems $d = \sum_{m=1}^M l_m$ allows for a hypothetical input space of size $D = \prod_{m=1}^M l_m \neq d$, in general. For instance, consider a \textit{maximal} decomposition into $\frac{d-1}{2}$ subsystems of size $2$ and one of size $3$. Applying our (hypothetical) correlation device we combine subsystems and find an input space of size $D := 3 \cdot 2^{\frac{d-1}{2}} \gneq d$, clearly larger than the original space. Evidently, $\Omega^{\mathbb{F}_D}_n \neq \Omega^{\mathbb{F}_{d}}_n$, we conclude that under additive composition of systems there would be a difference in viewing the system as composite or as sum of its parts. On the contrary, if we assume that physical subsystems combine multiplicatively with respect to their dimensions, we find upon otherwise similar reasoning, $D := p^r = d$ which exactly reproduces the full space of functions $g: \mathbb{F}^n_d \longrightarrow \mathbb{F}_d$ as required.

Hence, under the constraint that function spaces corresponding to a system are the same whether viewed as composite or in terms of its subsystems, we find that system dimensions need to scale multiplicatively. The tensor product is multiplicative in the dimension of the respective Hilbert spaces, quantum systems thus fit into the above argument. However, note that we have not assumed Hilbert spaces or any other related structures. The same argument therefore applies to all correlated systems independent of their physical implementation under local universality.


\subsection{Temporal Ordering}\label{sec: Temporal Ordering}

In this paper we have restricted the discussion to temporally flat MBQCs.  That is, the measurement settings for the $k$th qudit depend only on the input $\mathbf{i}$ to the control computer and not previous measurement outcomes from qudits $k'<k$. Temporal flatness turns out to be a crucial requirement, as we have seen that non-locality is the key resource behind computational speed-up. Allowing temporal ordering means to allow for (one-way) communication between the locally separated qudit sites. The constraints in Thm.~\ref{thm: generalised correspondence contextuality and computational complexity} and Cor.~\ref{cor: generalised correspondence contextuality and computational complexity} can be understood as Bell-like constraints on correlations for which communication is naturally excluded.

In fact, if we allowed for temporal ordering, we would have access to recursive function calls. It is straightforward to show that a classification of function spaces under iterative function calls with linear processing only has two stable subspaces, the entire space $\Omega^{\mathbb{F}_d}_n$ and the space of linear functions $L^{\mathbb{F}_d}_n$. This is why temporal ordering can be allowed in the qubit case. On the other hand, it means that any non-linear function $\phi: \mathbb{F}_d \longrightarrow \mathbb{F}_d$ elevates the control computer to arbitrary functions in $\Omega^{\mathbb{F}_d}_n$ under linear processing.

In particular, we can obtain an upper bound on the polynomial degree for non-contextual $ld$-MBQCs even in the temporally ordered case, but it will be larger in general. We first define a directed graph $\mathcal{G}$ that contains all of the information of the temporal ordering. Namely, we have a vertex for each party (qudit), and an edge whenever the choice of measurement on one party depends on the measurement result of another. For the computation to be executable, we require that there are no cycles in the graph. As such, the graph for the temporally flat case contains no edges.

To find an upper bound on the degree of a computed polynomial in a non-contextual $ld$-MBQC, we find the longest (directed) path $l$ in the graph $\mathcal{G}$. Let the number of vertices in this path be denoted $|l|$. Then by composing a degree $d-1$ polynomial $|l|$ times, we obtain a polynomial of degree $(d-1)^{|l|}$. Whenever an $ld$-MBQC is evaluating a function that is a polynomial with degree greater than $(d-1)^{|l|}$, we have a proof of strong contextuality. Thus with temporal ordering, it is more difficult to find proofs of strong contextuality within the setting of $ld$-MBQC.


\section{Discussion}\label{sec: Discussion}

In summary, we have placed a bound on the space of output functions of general MBQCs if the underlying system is non-contextual.  This generalises Raussendorf's result~\cite{Raussendorf2009} for qubits.  Nontrivial MBQCs on qudits do not directly make use of (local) contextuality (in quantum systems of dimension at least three), but instead harness a type of global contextuality, namely strong non-locality. In fact, Thm.~\ref{thm: generalised correspondence contextuality and computational complexity} can be understood as a deterministic version of Bell's Theorem restricted to MBQC: assuming local hidden variables, $m_k(x)$, and local (single qudit) measurements, $\Omega^{\mathbb{F}_d}_n(d-1)$ is already maximal.

Our results highlight strong non-locality as the crucial ingredient to Thm.~\ref{thm: generalised correspondence contextuality and computational complexity}, and show how non-local correlations between spatially separated subsystems can boost the computational power of the classical control computer.

We conclude with a number of directions for further research in this area. 

A key open question is whether our results (specifically, Cor.~\ref{cor: generalised correspondence contextuality and computational complexity}) generalise to qudits with arbitrary non-prime-power dimensions. Most of our proofs rely on the theorem regarding polynomials in finite fields, which says that all functions over finite fields are  polynomials.  However, this result no longer holds true if $d \neq p^r$ for $p$ prime and $r \in \mathbb{N}$, and so a generalisation would need to consider more general functions.

We have seen that the standard framework for MBQC in quantum theory is locally universal, and due to non-local correlations they can even compute functions outside of $\Omega^{\mathbb{F}_d}_n(d-1)$. However, it is not clear whether local universality allows us to compute all functions in $\Omega^{\mathbb{F}_d}_n$.  It remains to determine if MBQCs of this form in quantum theory can saturate the space of functions (perhaps focussing on prime power dimensions).  A more detailed analysis of the connections between the computational results and multi-party Bell inequalities as derived in Ref.~\cite{HobanWallmanBrowne2011} provides a natural starting point.

We focussed on generalised Pauli-like measurements with Clifford control unitaries as our framework, as this is standard for quantum computation, and showed that it is locally universal.  With a different choice of control unitaries, it might be possible to exclude local universality. In this case, one may be able to obtain a proof of contextuality with the computation of polynomial with degree lower than $d-1$.

Contextuality has recently been related to cohomology~\cite{AbramskyBarbosaMansfield2012,raussendorf2016cohomological,okay2017topological,OkayTyhurstRaussendorf2018,Caru2018}. It would be a worthwhile goal to pursue such classifications further and bring them in contact with our result on computational complexity.  It would be interesting to understand the relationship between polynomial degree in our setting, and the nontriviality of a certain cocycle within the cohomological frameworks of Refs.~\cite{raussendorf2016cohomological,okay2017topological}.

\begin{acknowledgments}
We thank Robert Spekkens and Joel Wallman for discussions. This work is supported by the Australian Research Council (ARC) via the Centre of Excellence in Engineered Quantum Systems (EQuS) project number CE170100009 and through a studentship in the Centre for Doctoral Training on Controlled Quantum Dynamics at Imperial College London funded by the EPSRC. SR also acknowledges support from the Australian Institute for Nanoscale Science and Technology Postgraduate Scholarship (John Makepeace Bennett Gift).
\end{acknowledgments}

\newpage

\appendix


\section{Proof of Theorem on Functions on Finite Fields}\label{secAppA}

Let $g: \mathbb{F}_d^n \rightarrow \mathbb{F}_d$. We first show how to obtain the $\delta$-function from a polynomial with partial degree less than or equal to $d-1$. For any $y \in  \mathbb{F}_d^n$ consider the delta function defined as 
\begin{equation}
\delta(x-y) = \begin{cases}
1 \quad \text{if } x = y \\
0 \quad \text{otherwise.}
\end{cases}
\end{equation}
Then it holds that 
\begin{equation}\label{eqPolyDelta}
\delta(x-y) = \prod_{i=1}^n (1 - (x_i-y_i)^{d-1}),
\end{equation}
which follows from Fermat's little theorem for $d$ prime and for general finite fields as every element in the multiplicative group $\mathbb{F}_d^\times$ has order a divisor of $d-1$. We can express any function $g: \mathbb{F}_d^n \rightarrow \mathbb{F}_d$ as a linear combination of delta functions,
\begin{equation}
g(x) = \sum_{y \in \mathbb{F}_d^n} \delta (x - y) \phi(y).
\end{equation}
which along with Eq.~(\ref{eqPolyDelta}), completes the proof.


\section{Proof of Theorem 2}\label{secAppB}

\begin{proof}[Proof of Thm.~\ref{thm: local universality in MBQC}]
Lem.~\ref{lm: linearly stable subspaces} says a polynomial $\phi$ of degree $|\phi| = l$ generates all functions of less or equal degree via linear pre- and post-processing. The higher its degree the more valuable a function becomes when viewed as a resource for computational tasks under the scheme in Fig.~\ref{fig: schematic}. It is thus sufficient to show that we can implement a function of highest degree within the limits of $ld$-MBQC. Here we will make use of the Clifford control unitary explored in the example of Sec.~\ref{sec: Example 3:  more non-linear functions}.

Instead of proving the existence of a function of degree $p-1$, $p$ prime we take a slightly different route and consider the $\delta(x)$-function over the finite field $\mathbb{Z}_p$.
\begin{table}[h!]
\centering
\begin{tabular}{c|cccc}
  & 0 & 1 & $\cdots$ & $p-1$ \\
\hline
$\delta(x)$ & 1 & 0 & $\cdots$ & 0
\end{tabular}
\end{table}
It is clear that with $\delta(x)$ as a resource we can generate in particular any polynomial $m \in \mathbb{Z}_p[x]$ via linear combinations of the form,
\begin{equation*}
m(x) = \sum_{j=0}^{p-1} m_j\delta(x-j) \quad ,
\end{equation*}
where $m_j \in \mathbb{Z}_p$ are constants specific to $m$.

In all finite fields the sum of elements within a subgroup of the multiplicative group $\mathbb{F}_d^\times$ vanishes as,
\begin{equation*}
S := \sum_{g \in G} g = \sum_{h^{-1}g \in G} g = h \sum_{g' \in G} g' = hS.
\end{equation*}
Observe, that the exponential function $u^x$ maps the additive group $\mathbb{Z}_p$ into its multiplicative group $\mathbb{Z}_p^\times$. We assume $u$ to be a primitive element of $\mathbb{Z}_p^\times$, then $u^x$ is surjective and we obtain non-trivial subgroups $\langle u^x \rangle \subseteq \mathbb{Z}_p^\times$, for all $x \neq 0,p-1$. One computes,
\begin{align}\label{eq: sum over exponents}
\sum_{k=1}^{p-1} (u^x)^k &= \frac{|\mathbb{Z}_p^\times|}{|\langle u^x \rangle|} \ \sum_{k=0}^{|\langle u^x \rangle| - 1} (u^x)^k \nonumber \\
&= \frac{p-1}{|\langle u^x \rangle|} \ \sum_{g \in \langle u^x \rangle} g = \ \begin{cases}
p-1 &\mathrm{if} \ x=0,p-1, \\
0 &\mathrm{otherwise}, \\
\end{cases}
\end{align}
as the sum either runs over all non-zero group elements in $\mathbb{Z}_p$ in case $u^x$ is a primitive of $\mathbb{Z}_p^\times$, or over the subgroup $\langle u^x \rangle \subseteq \mathbb{Z}_p^\times$ for $x \neq 0,p-1$.

As shown in the example of Sec.~\ref{sec: Example 3:  more non-linear functions}, we can obtain this exponential function using the symplectic Clifford unitary $M_u$ of Eq.~(\ref{eq:Mr}). Using Eq.~(\ref{eq: sum over exponents}) and the exponential function constructed in Eq.~(\ref{eq: construct exponential}) with linear pre- and post-processing on $p-1$ qudits we implement the function,
\begin{equation}\label{eq: construct delta}
\sigma_p(x) := (p-1)^{-1} \sum_{l=1}^{p-1} \ (u^x)^l =
\begin{cases}
1 \ &\mathrm{if} \ x = 0,p-1, \\
0 \ &\mathrm{otherwise}.
\end{cases} 
\end{equation}
For $p=5$ and $r=2$, the exponentials and their sum in Eq.~(\ref{eq: construct delta}) are explicitly given in the table below
\begin{table}[h!]
\centering
\begin{tabular}{c|ccccc}
$\mathit{x}$ & 0 & 1 & 2 & 3 & 4 \\
\hline
$2^{x}$ & 1 & 2 & 4 & 3 & 1 \\
$2^{2x}$ & 1 & 4 & 1 & 4 & 1 \\
$2^{3x}$ & 1 & 3 & 4 & 2 & 1 \\
$2^{4x}$ & 1 & 1 & 1 & 1 & 1 \\
\hline
$\sigma_5(x)$ & 1 & 0 & 0 & 0 & 1 \\
\end{tabular}
\label{tab: p=5}
\end{table}

Finally, we use another linear operation to recover $\delta(x)$ from $\sigma_p(x)$,
\begin{equation}
\delta(x) = 1+ \frac{p-1}{2}(1+ \sum_{k=1}^{p-1} \sigma_p(kx)).
\end{equation}
In summary, we construct any function $m: \mathbb{Z}_p \longrightarrow \mathbb{Z}_p$ using the Clifford unitary $M_u$ on $p(p-1)^2$ qudits under linear pre- and post-processing (without temporal ordering),
\begin{align*}
m(x) &= \sum_{j=0}^{p-1} m_j \left(1+ \frac{p-1}{2}\left(1+ \sum_{k=1}^{p-1} \left((p-1)^{-1} \sum_{l=1}^{p-1} \ u^{lk(x-j)}\right)\right)\right) \nonumber \\
&= \frac{1}{2}\sum_{j=0}^{p-1} m_j \left(1 + \sum_{k=1}^{p-1} \sum_{l=1}^{p-1} \ u^{lk(x-j)}\right).\\
\end{align*}

Note that the proof holds for general finite fields as long as we can implement the exponential function, yet Eq.~(\ref{eq: construct exponential}) relies on the explicit $ld$-MBQC implementation and constrains to $\mathbb{Z}_p$ for $p$ prime. 

In the latter case we can improve the number of qudits needed to implement the $\delta$-function and generalise to qudits of arbitrary odd dimension $d$. Note that in any ring of integers modulo $d$,
\begin{equation*}
a^2 = 1 \ \mathrm{mod} \ d \Longleftrightarrow (a-1)(a+1) = 0 \ \mathrm{mod} \ d,
\end{equation*}
has solutions $a = \pm 1$. We can use this to construct,
\begin{equation*}
\delta(x) = \frac{1}{2} \sum_{k=\pm 1} (d-1)^{kx} =
\begin{cases}
1 &\mathrm{if} \ x=0 \\
0 & \mathrm{otherwise}
\end{cases}
\end{equation*}
and thus arbitrary functions $m: \mathbb{Z}_d \longrightarrow \mathbb{Z}_d$,
\begin{equation*}
m(x) = \sum_{j=0}^{p-1} m_j \left( \frac{1}{2} \sum_{k=\pm 1} (d-1)^{k(x-j)} \right),
\end{equation*}
using the Clifford unitary $M_{d-1}$ on $2p$ qudits under linear pre- and post-processing (without temporal ordering).
\end{proof}


\section{The Probabilistic Case}\label{secAppC}

Cor.~\ref{cor: generalised correspondence contextuality and computational complexity} holds in the deterministic case only, however, the result remains valid at least in a small neighbourhood around the success probability $p_S = 1$, where
\begin{equation*}
p_S := \min_{\mathbf{i} \in \mathbb{Z}_d^n} \ \mathrm{Prob}(\tau(\mathbf{i}) = o(\mathbf{i})),
\end{equation*}
$\tau(\mathbf{i})$ the factual result of a computation with input string $\mathbf{i} \in \mathbb{Z}^n_d$. This is an immediate generalisation of the analogous result for qubits in \cite{Raussendorf2009}, with some slight adjustments. First, we define $\Delta$ as follows,
\begin{align*}
\Delta: \mathbb{Z}_d &\longrightarrow \mathbb{Z}_{\frac{(d-1)}{2}} \nonumber \\
\Delta (q) &:= \begin{cases} q \ \ \ \mathrm{if} \ q < -q \\ -q \ \mathrm{if} \ q > -q \end{cases}.
\end{align*}

Next, we define the distance $\nu$ of a function $o: \mathbb{Z}^n_d \rightarrow \mathbb{Z}_d$ to the closest polynomial function in $\Omega^{\mathbb{Z}_d}_n(d-1)$,
\begin{equation}\label{eq: minimal function distance}
\nu(o) := \min_{p \in \Omega^{\mathbb{Z}_d}_n(d-1)} \ \sum_{\mathbf{i} \in \mathbb{Z}_d^n} \Delta (o(\mathbf{i}) - p(\mathbf{i})),
\end{equation}
where the minimum is taken over all polynomial functions of degree at most $d-1$. Generalising the qubit case \cite{Raussendorf2009} we observe strong non-locality for,
\begin{equation}
p_S > 1 - \frac{1}{d^n}\frac{2\nu(o)}{(d-1)}.\\
\end{equation}

In \cite{AbramskyBarbosaMansfield2017} the authors derive a stronger result based on the average success probability,
\begin{equation*}
\bar{p}_S := \frac{1}{d^n} \sum_{\mathbf{i} \in \mathbb{Z}_d^n} \ \mathrm{Prob}(\tau(\mathbf{i}) = o(\mathbf{i})).
\end{equation*}
Decomposing the empirical model $e$ (a set of probability distributions for outcomes and experiments corresponding to different contexts, cf. \cite{AbramskyBrandenburger2011}) underlying the $ld$-MBQC into a non-contextual and strongly contextual part,
\begin{equation*}
e = \mathrm{NCF}(e)e^\mathrm{NC} + \mathrm{CF}(e)e^\mathrm{SC},
\end{equation*}
where $\mathrm{NCF}(e) \in [0,1]$, $\mathrm{CF}(e) = 1 - \mathrm{NCF}(e)$,
the average probability of failure, $\bar{p}_F = 1 - \bar{p}_S$, satisfies the relation,
\begin{equation*}
\bar{p}_F \geq \mathrm{NFC}(e)\nu(o).
\end{equation*}
A violation of this inequality yields a proof of non-locality and thus generalises the result achieved in \cite{AbramskyBarbosaMansfield2017} where only linear functions were considered. The analogous derivation holds under the appropriate change in Eq.~(\ref{eq: minimal function distance}).

Note that the success probability depends on the non-contextual fraction $\mathrm{NFC}(e)$, which measures the amount of non-contextuality inherent in the probabilistic case. On the other hand it also depends on $\nu(o)$, the minimal distance of the output function to a function realisable in the non-contextual case. We could thus achieve a better success probability if we restricted to functions that all have a certain minimal distance from the set of possible output functions. However, as the space of output functions under non-contextuality, $\Omega^{\mathbb{Z}_d}_n(p-1)$, is non-linear, we cannot use bent functions as in the qubit case.
\end{document}